\documentclass[11pt,letterpaper]{article}

\usepackage{amsmath,amsfonts,amsthm,amssymb,stmaryrd,graphicx,relsize,bm}  % 9.20 add graphicx,resize
\theoremstyle{plain}

\numberwithin{equation}{section}

\newtheorem{thm}{Theorem}[section]
\newtheorem{lem}[thm]{Lemma}

\newenvironment{exam}[1]
{\begin{flushleft}\textbf{Example #1}.\enspace}%
{\end{flushleft}}

\allowdisplaybreaks  % introduced 7.15.15

\newcommand{\real}{{\mathbb R}}

\newcommand{\trace}{tr}
\newcommand{\rmcor}{\mathrm{Cor}}
\newcommand{\rmcomm}{\mathrm{Comm\,}}
\newcommand{\instr}{In}
\newcommand{\ob}{Ob}
\newcommand{\ityes}{\textit{yes}}
\newcommand{\itno}{\textit{no}}
\newcommand{\rmre}{\mathrm{Re}}

\newcommand{\escript}{\mathcal{E}}

\newcommand{\iscript}{\mathcal{I}}
\newcommand{\jscript}{\mathcal{J}}
\newcommand{\lscript}{\mathcal{L}}
\newcommand{\mscript}{\mathcal{M}}
\newcommand{\oscript}{\mathcal{O}}

\newcommand{\sscript}{\mathcal{S}}

\newcommand{\alphahat}{\widehat{\alpha}}
\newcommand{\iscripthat}{\widehat{\iscript}}
\newcommand{\jscripthat}{\widehat{\jscript}}

\newcommand{\atilde}{\widetilde{A}}

\newcommand{\ptilde}{\widetilde{P}}
\newcommand{\iscripttilde}{\widetilde{\iscript}}
\newcommand{\mscripttilde}{\widetilde{\mscript}}
\newcommand{\iscriptbar}{\overline{\iscript}}

\newcommand{\alphabar}{\overline{\alpha}}

\newcommand{\ab}[1]{\left|#1\right|}
\newcommand{\doubleab}[1]{\left|\left|#1\right|\right|}
\newcommand{\brac}[1]{\left\{#1\right\}}
\newcommand{\paren}[1]{\left(#1\right)}
\newcommand{\sqbrac}[1]{\left[#1\right]}

\begin{document}

\title{MEASUREMENT MODELS\\WITH SEPARABLE\\INTERACTION CHANNELS}
\author{Stan Gudder\\ Department of Mathematics\\
University of Denver\\ Denver, Colorado 80208\\
sgudder@du.edu}
\date{}
\maketitle

\begin{abstract}
Measurement models (MMs) stand at the highest structural level of quantum measurement theory. MMs can be employed to construct instruments which stand at the next level. An instrument is thought of as an apparatus that is used to measure observables and update states. Observables, which are still at the next level, are used to determine probabilities of quantum events. The main ingredient of an MM is an interaction channel $\nu$ between the system being measured and a probe system. For a general $\nu$, the measured observable $A$ need not have an explicit useful form. In this work we introduce a condition for $\nu$ called separability and in this case $A$ has an explicit form. Under the assumption that $\nu$ is separable, we study product MMs and conditioned MMs. We also consider the statistics of MMs and their uncertainty principle. Various concepts are illustrated using examples of L\"uders and Holevo instruments.
\end{abstract}

\section{Introduction}  % Section 1
For simplicity, we shall assume that all our Hilbert spaces are finite dimensional. Although this is a strong restriction, these spaces are general enough to include quantum computation and information theory \cite{hz12,nc00}. Moreover, many of our results can be generalized to the infinite dimensional case. We begin with a general discussion of the material and leave many detailed definitions and results for later sections. If $H$ is a finite dimensional complex Hilbert spaces, we denote the set of (bounded) linear operators on $H$ by $\lscript (H)$ and the set of self-adjoint operators on $H$ by $\lscript _S(H)$. A positive operator $\rho\in\lscript _S(H)$ with trace $\trace (\rho )=1$ is called a \textit{state}. States are employed to describe the condition of a quantum system and the set of states is denoted by $\sscript (H)$. In particular, if $\rho\in\sscript (H)$, then $\rho$ can be used to compute the statistics of the system. As an example, an operator $a\in\lscript _S(H)$ is called an
\textit{effect} if $0\le a\le I$ where $0$ and $I$ are the zero and identity operators of $H$, respectively \cite{bgl95,blm96,blpy16,hz12}. Effects correspond to two valued\ \ityes-\itno\ (true-false) experiments and the set of effects is denoted by $\escript (H)$. If $a\in\escript (H)$ and $\rho\in\sscript (H)$, then the \textit{probability} that $a$ has the value \ityes\  ($a$ is true) when the system is in the state $\rho$ is $\trace (\rho a)$. Since $0\le a\le I$, it follows that $0\le\trace (\rho a)\le 1$. Moreover, the complementary effect $a'=I-a$ is true if and only if $a$ is false and since $a+a'=I$ we have $\trace (\rho a')=1-\trace (\rho a)$.

One of the most important problems in quantum mechanics is to determine the state of a quantum system by performing measurements on the system and this is where measurement models (MMs) come in. Suppose the system is in the unknown state $\rho\in\sscript (H)$ and we wish to gain information about $\rho$. To accomplish this, we first interact $H$ with an auxiliary system whose Hilbert space is $K$. The composite of the two systems is given by the combined tensor product system $H\otimes K$ and the interaction is described by a quantum channel $\nu\colon\sscript (H)\to\sscript (H\otimes K)$ \cite{bgl95,blm96,blpy16,oza84}. A measurement of a probe observable $P$ on $K$ is performed and the result gives information about $\rho$. To be precise, as we shall see in Section~2, an apparatus called an instrument is applied and this measures an observable $A$ whose probability distribution describes $\rho$. The MM for this process is given by the triple
$\mscript =(K,\nu ,P)$.

Recently MMs have been extensively studied in the literature \cite{bgl95,blm96,blpy16,hz12,lp22,oza84}. Unfortunately, general interaction channels can be quite complicated and difficult to analyze. In particular the measured observable $A$ need not have an explicit useful form. In this work we introduce a condition on $\nu$ called separability and in this case $A$ has an explicit form. In later sections we study conditioned MMs and sequential product MMs under the assumption that $\nu$ is separable. Moreover, we consider the statistics of real-valued MMs and their uncertainty principle. In particular, we define the expectation and variance of a real-valued MM. We also define the correlation and covariance between two MMs. For this work, an important role is played by the stochastic operator of a MM. Various concepts are illustrated using examples of L\"uders and Holevo instruments.

\section{Basic Definitions}  % Section 2
A (finite) \textit{observable} on a Hilbert space $H$ is a finite set $A=\brac{A_x\colon x\in\Omega _A}$ of effects $A_x\in\escript (H)$ that satisfy
$\sum\limits _{x\in\Omega _A}A_x=I$ \cite{bgl95,blm96,blpy16,hz12,lp22,oza84}. We call $\Omega _A$ the \textit{outcome set} of $A$ and the elements
$x\in\Omega _A$ the \textit{outcomes} of $A$. When $A$ is measured and the resulting outcome is $x$, we say that $A_x$ \textit{occurs}. The set of observables on $H$ is denoted $\ob (H)$. If $\rho\in\sscript (H)$, $A\in\ob (H)$, the $\rho$-\textit{probability distribution} of $A$ is $\Phi _\rho ^A(x)=\trace (\rho A_x)$. Notice that
$\Phi _\rho ^A$ is indeed a probability measure because $0\le\trace (\rho A_x)\le 1$ for all $x\in\Omega _A$ and 
\begin{equation*}
\sum _{x\in\Omega _A}\Phi _\rho ^A(x)=\sum _{x\in\Omega _a}\trace (\rho A_x)=\trace\sqbrac{\rho\sum _{x\in\Omega _A}A_x}
   =\trace (\rho I)=\trace (\rho )=1
\end{equation*}
We conclude that an observable is an effect-valued measure and if $\Delta\subseteq\Omega _A$ we call $\Delta$ an \textit{event} and write
$A(\Delta )=\sum _{x\in\Delta}A_x$. Observables are also called \textit{positive operator-valued measures} (POVM) \cite{blpy16,hz12,nc00}.

An \textit{operation} from $H$ to another Hilbert space $K$ is a completely positive map \cite{bgl95,blm96,blpy16,hz12,lp22} $\iscript\colon\lscript (H)\to\lscript (K)$ such that $\trace\sqbrac{\iscript (B)}\le\trace (B)$ for all $B\in\lscript (H)$. We denote the set of operations from $H$ to $K$ by $\oscript (H,K)$ and write
$\oscript (H)=\oscript (H,H)$. If $\iscript\in\oscript (H,K)$, there exists a finite set of \textit{Kraus operators} $C_i\colon H\to K$ such that
$\sum\limits _{i=1}^nC_i^*C_i\le I$ and $\iscript (B)=\sum\limits _{i=1}^nC_iBC_i^*$ for all $B\in\lscript (H)$ \cite{hz12,lp22,nc00}. If $\iscript\in\oscript (H,K)$ there exists an unique linear map $\iscript ^*\colon\lscript (K)\to\lscript (H)$ \cite{gud120,gud220} such that
$\trace\sqbrac{\rho\iscript ^*(B)}=\trace\sqbrac{\iscript (\rho )B}$. In fact, if $C_i\colon H\to K$ are Kraus operators for $\iscript$, then
$\iscript ^*(B)=\sum\limits _{i=1}^nC_i^*BC_i$. We call $\iscript ^*$ the \textit{dual} of $\iscript$. Notice that $\iscript ^*\colon\escript (K)\to\escript (H)$. An operation
$\iscript$ that satisfies $\trace\sqbrac{\iscript (B)}=\trace (B)$ for all $B\in\lscript (H)$ is called a \textit{channel}. If $\iscript$ is an operation with Kraus operators
$C_i\colon H\to K$, it is easy to check that $\iscript$ is a channel if and only if $\sum\limits _{i=1}^nC_i^*C_i=I$. It follows that $\iscript$ is a channel if and only if
$\iscript ^*(I)=I$.

An \textit{instrument} is a finite set of operations $\iscript =\brac{\iscript _x\colon x\in\Omega _\iscript}\subseteq\oscript (H)$ where
$\iscriptbar =\sum\limits _{x\in\Omega _\iscript}\iscript _x$ is a channel. We call $\Omega _\iscript$ the \textit{outcome set} of $\iscript$ and denote the set of instruments on $H$ by $\instr (H)$.  If $\rho\in\sscript (H)$, $\iscript\in\instr (H)$, the $\rho$-probability distribution of $\iscript$ is
$\Phi _\rho ^\iscript (x)=\trace\sqbrac{\iscript _x(\rho )}$. As with observables, $\Phi _\rho ^\iscript$ is a probability measure because
$0\le\trace\sqbrac{\iscript _x(\rho )}\le 1$ for all $x\in\Omega _\iscript$ and 
\begin{equation*}
\sum _{x\in\Omega _\iscript}\Phi _\rho ^\iscript (x)=\sum _{x\in\Omega _\iscript}\trace\sqbrac{\iscript _x(\rho )}=\trace\sqbrac{\iscriptbar (\rho )}=\trace (\rho )=1
\end{equation*}
An instrument is thought of as an apparatus that performs a measurement and updates the state of the system depending on the result. If the apparatus for $\iscript$ is executed and the result $x\in\Omega _\iscript$ occurs, then the \textit{updated} state is $\iscript _x(\rho )/\doubleab{\iscript _x(\rho )}$ whenever
$\iscript _x(\rho )\ne 0$. As with observables, an instrument is an operation-valued measure and for $\Delta\subseteq\Omega _\iscript$ we write
$\iscript (\Delta )=\sum\limits _{x\in\Delta}\iscript _x$.

A \textit{stochastic kernel} (or \textit{stochastic matrix}) is a collection $\brac{\lambda _{yx}\colon x\in\Delta ,y\in\Gamma}$ where $\Delta ,\Gamma$ are finite sets,
$0\le\lambda _{yx}\le 1$ and $\sum\limits _{x\in\Gamma}\lambda _{yx}=1$, for $y\in\Gamma$. If $\lambda _{yx}$ is a stochastic kernel and $\iscript\in\instr (H)$, then
$\jscript =\brac{\jscript _x\colon x\in\Delta}$ defined by $\jscript _x=\sum _{y\in\Gamma}\lambda _{yx}\iscript _y$ is an instrument called a \textit{post-processing} of
$\iscript$. Similarly, if $A$ is an observable, then $B=\brac{B_x\colon x\in\Delta}$ given by $B_x=\sum\limits _{y\in\Gamma}\lambda _{yx}A_y$ is an observable called a \textit{post-processing} of $A$ \cite{blpy16,hz12}. If $\iscript\in\instr (H)$, then there exists an unique observable $\iscripthat\in\ob (H)$ such that
$\trace (\rho\iscripthat _x)=\trace\sqbrac{\iscript _x(\rho )}$ for all $x\in\Omega _{\iscripthat}=\Omega _\iscript$, $\rho\in\sscript (H)$. We call $\iscripthat$ the observable \textit{measured} by $\iscript$ \cite{gud120,gud220}. Notice that $\iscript$ and $\iscripthat$ have the same $\rho$-probability distribution for all
$\rho\in\sscript (H)$. However, $\iscript$ gives more information than $\iscripthat$ because $\iscript$ also determines the updated state. Although an instrument measures an unique observable as we shall see, an observable is measured by many instruments. Moreover, if
$\jscript _x=\sum\limits _{y\in\Gamma}\lambda _{yx}\iscript _y$ is a post-processing of $\iscript$, and it is easily shown that
$\iscripthat _x=\sum\limits _{y\in\Gamma}\lambda _{yx}\iscripthat _x$, it follows that $\jscripthat$ is a post-processing of $\iscripthat$. Since
$\trace (\rho\iscripthat _x)=\trace\sqbrac{\iscript _x(\rho )I}=\trace\sqbrac{\rho\iscript _x^*(I)}$, we see that $\iscripthat _x=\iscript _x^*(I)$.

A \textit{measurement model} (MM) for a Hilbert space $H$ is a triple $\mscript =(K,\nu ,P)$ where $K$ is an auxiliary Hilbert space,
$\nu\colon\sscript (H)\to\sscript (H\otimes K)$ is a (interaction) channel and $P$ is a (probe) observable on $K$. This is slightly different but equivalent to the usual definition \cite{bgl95,blm96,blpy16,hz12,oza84}. The instrument \textit{measured} by $\mscript$ is defined by 
\begin{equation*}
\iscript _x^\mscript (\rho )=\trace _K\sqbrac{\nu (\rho )I\otimes P_x}
\end{equation*}
for all $x\in\Omega _{\iscript ^\mscript}=\Omega _P$, $\rho\in\sscript (H)$ where $\trace _K$ is the partial trace relative to $K$. The channel for $\iscript ^\mscript$ is
\begin{equation*}
\iscriptbar ^{\,\mscript}(\rho )=\sum _x\iscript _x^\mscript (\rho )=\trace _K\sqbrac{\nu (\rho )}
\end{equation*}
In general, there is no explicit expression for the observable $(\iscript ^\mscript )^\wedge$ measured by $\iscript ^\mscript$. All we can do is write
\begin{equation*}
\trace\sqbrac{\rho (\iscript ^\mscript )_x^\wedge}=\trace\sqbrac{\iscript _x^\mscript (\rho )}=\trace\sqbrac{\nu (\rho )I\otimes P_x}
   =\trace\sqbrac{\rho\nu ^*(I\otimes P_x)}
\end{equation*}
from which we have $(\iscript ^\mscript )_x^\wedge =\nu ^*(I\otimes P_x)$ but this gives no useful information, in general. This problem is overcome in the case where $\nu$ is \textit{separable} which means that for all $\rho\in\sscript (H)$
\begin{equation}                % equation (2.1)
\label{eq21}
\nu (\rho )=\sum _y\alpha _y(\rho )\otimes\gamma _y
\end{equation}
where $\gamma _y\in\sscript (K)$ and $\alpha =\brac{\alpha _y\colon y\in\Omega _\alpha}\in\instr (H)$.

\begin{lem}    % Lemma 2.1
\label{lem21}
If $\nu =\sum\limits _y\alpha _y(\rho )\otimes\gamma _y$is separable, then
\begin{equation}                % equation (2.2)
\label{eq22}
\iscript _x^\mscript =\sum _y\trace (\gamma _yP_x)\alpha _y
\end{equation}
is a post-processing of $\alpha$ and
\begin{equation}                % equation (2.3)
\label{eq23}
(\iscript ^\mscript )_x^\wedge =\sum _y\trace (\gamma _yP_x)\alphahat _y
\end{equation}
is a post-processing of $\alphahat$.
\end{lem}
\begin{proof}
For every $\rho\in\sscript (H)$ we obtain
\begin{align*}
\iscript _x^\mscript (\rho )&=\trace\sqbrac{\sum _y\alpha _y(\rho )\otimes\gamma _yI\otimes P_x}=\trace _K\sqbrac{\sum _y\alpha _y(\rho )\otimes\gamma _yP_x}\\
   &=\sum _y\trace _K\sqbrac{\alpha _y(\rho )\otimes\gamma _yP_x}=\sum _y\trace (\gamma _yP_x)\alpha _y(\rho )
\end{align*}
from which \eqref{eq22} follows. To find $(\iscript ^\mscript )_x^\wedge$ we have from \eqref{eq22} that
\begin{align*}
(\iscript ^\mscript )_x^\wedge&=(\iscript _x^\mscript )^*(I)=\sqbrac{\sum _y\trace (\gamma _yP_x)\alpha _y}^*(I)=\sum _y\trace (\gamma _yP_x)\alpha _y^*(I)\\
   &=\sum _y\trace (\gamma _yP_x)\alphahat _y\qedhere
\end{align*}
\end{proof}

\begin{exam}{1}  % Example 1
A \textit{Kraus instrument} has the form $\alpha _y(\rho )=C_y\rho C_y^*$ where $C_y\in\lscript (H)$ with $\sum\limits _yC_y^*C_y=I$. We then have that
$\alpha _y^*(B )=C_y^*BC_y$ for all $B\in\lscript (H)$ and $\alphahat _y=\alpha _y^*(I)=C_y^*C_y$. If $\alpha$ is a Kraus instrument as just described and $\nu$ is separable with form \eqref{eq21}, then $\nu (\rho )=\sum\limits _yC_y\rho C_y^*\otimes\gamma _y$ and by \eqref{eq22} and \eqref{eq23} we obtain
\begin{equation*}
\iscript _x^\mscript (\rho )=\sum _y\trace (\gamma _yP_x)C_y\rho C_y^*,\quad (\iscript ^\mscript )_x^\wedge =\sum _y\trace (\gamma _yP_x)C_y^*C_y
\end{equation*}
A special case of a Kraus instrument is a L\"uders instrument which has the form $\alpha _y(\rho )=A_y^{1/2}\rho A_y^{1/2}$ where $A_y\in\escript (H)$ and
$\sum _yA_y=I$. Thus, $A=\brac{A_y\colon y\in\Omega _A}$ is an observable, $\alpha _y^*(B)=A_y^{1/2}BA_y^{1/2}$ and $\alphahat =A$ \cite{gud120,gud220}. We then obtain 
\begin{equation*}
\iscript _x^\mscript (\rho )=\sum _y\trace (\gamma _yP_x)A_y^{1/2}\rho A_y^{1/2},\quad
   (\iscript ^\mscript )_x^\wedge =\sum _y\trace (\gamma _yP_x)A_y\hskip 3pc\square
\end{equation*}
\end{exam}

\begin{exam}{2}  % Example 2
If $A=\brac{A_y\colon y\in\Omega _A}\in\oscript (H)$, $\beta _y\in\sscript (H)$, we define the \textit{Holevo instrument} $\alpha\in\instr (H)$ by
$\alpha _y(\rho )=\trace (\rho A_y)\beta _y$ \cite{gud120,gud220}. Since
\begin{equation*}
\trace\sqbrac{\rho\alpha _y^*(B)}=\trace\sqbrac{\alpha _y(\rho )B}=\trace (\rho A_y)\trace (\beta _yB)=\trace\sqbrac{\rho\trace (\beta _yB)A_y}
\end{equation*}
we have $\alpha _y^*(B)=\trace (\beta _yB)A_y$ and $\alphahat _y=A_y$. This shows that an observable can be measured by many instruments. If $\nu$ is separable with the form \eqref{eq21}, then \eqref{eq22} and \eqref{eq23} give
\begin{equation*}
\iscript _x^\mscript (\rho )=\sum _y\trace (\gamma _yP_x)\trace (\rho A_y)\beta _y,\quad (\iscript ^\mscript )_x^\wedge =\sum _y\trace (\gamma _yP_x)A_y
\end{equation*}
Notice that this observable is the same as the observable of Example~1. As a special case of a Holevo instrument, let $A_y=\lambda _yI$ where
$0\le\lambda _y\in 1$, $\sum\limits _y\lambda _y=1$. We call $A=\brac{A_y\colon y\in\Omega _A}$ an \textit{identity observable}. For $\beta\in\sscript (H)$, the Holevo instrument
\begin{equation*}
\alpha _y(\rho )=\trace (\rho A_y)\beta =\lambda _y\beta
\end{equation*}
is called a \textit{trivial instrument} \cite{gud120,gud220}. We then have $\alpha _y^*(B)=\lambda _y\trace (\beta B)I$ and $\alphahat _y=\lambda _yI$. If $\nu$ is separable with form \eqref{eq21}, then 
\begin{equation*}
\nu (\rho )=\sum _y\lambda _y\beta\otimes\gamma _y=\beta\otimes\gamma
\end{equation*}
where $\gamma =\sum\limits _y\lambda _y\gamma _y\in\sscript (K)$ and by \eqref{eq22}
\begin{equation*}
\iscript _x^\mscript (\rho )=\sum _y\lambda _y\trace (\gamma _yP_x)\beta =\trace (\gamma P_x)\beta
\end{equation*}
which is a trivial instrument. Moreover, \eqref{eq23} gives
\begin{equation*}
(\iscript ^\mscript )_x^\wedge =\sum _y\trace (\gamma _yP_x)\lambda _yI=\trace (\gamma P_x)I
\end{equation*}
which is an identity observable.\hfill\qedsymbol
\end{exam}

\section{Product and Conditioned Measurement Models}  % Section 3
If $\iscript ,\jscript\in\instr (H)$, we define the \textit{sequential product} $\iscript\circ\jscript$ with outcome space $\Omega _\iscript\times\Omega _\jscript$ given by
$(\iscript\circ\jscript )_{(x,y)}(\rho )=\jscript _y\sqbrac{\iscript _x(\rho )}$. We interpret $\iscript\circ\jscript$ as the instrument obtained by first measuring $\iscript$ and then measuring $\jscript$. The same definition applies to dual instruments and it can be shown \cite{gud120,gud220} that 
$(\iscript\circ\jscript )^*=\jscript ^*\circ\iscript^*$. It follows that
\begin{equation}                % equation (3.1)
\label{eq31}
(\iscript\circ\jscript )_{(x,y)}^\wedge =(\iscript\circ\jscript )_{(x,y)}^*(I)=(\jscript _y^*\circ\iscript _x^*)(I)=\iscript _x^*\paren{\jscript _y^*(I)}=\iscript _x^*(\,\jscripthat _y)
\end{equation}
Letting $\mscript =(K,\nu ,P)$, $\mscript '=\brac{K',\nu ',P'}$ be MMs for $H$, we have for all $\rho\in\sscript (H)$ that 
\begin{align}% equation (3.2)
\label{eq32}
(\iscript ^\mscript\circ\iscript ^{\mscript'})_{(x,y)}(\rho )&=\iscript _y^{\mscript'}\sqbrac{\iscript _x^\mscript (\rho )}
   =\iscript _y^{\mscript '}\brac{\trace _K\sqbrac{\nu (\rho )I\otimes P_x}}\notag\\
   &=\trace _{K'}\brac{\nu '\sqbrac{\trace _K\paren{\nu (\rho )I\otimes P_x}}I\otimes P'_y}
\end{align}

\begin{thm}    % Theorem 3.1
\label{thm31}
If $\nu (\rho )=\sum\limits _y\alpha _y(\rho )\otimes\gamma _y$, $\nu '(\rho)=\sum\limits _{y'}\alpha '_{y'}(\rho )\otimes\gamma '_{y'}$ are separable, then
\begin{align}% equation (3.3)
\label{eq33}
(\iscript ^\mscript\circ\iscript ^{\mscript '})_{(x,y)}&=\sum _{x',y'}\trace (\gamma _{x'}P_x)\trace (\gamma '_{y'}P'_y)(\alpha\circ\alpha ')_{(x',y')}
\intertext{and}
% equation (3.4)
\label{eq34}
(\iscript ^\mscript\circ\iscript ^{\mscript '})_{(x,y)}^\wedge&=\sum _{x',y'}\trace (\gamma _{x'}P_x)\trace (\gamma '_{y'}P '_y)\alpha _{x'}^*\sqbrac{(\alpha ')_{y'}^\wedge}
\end{align}
\end{thm}
\begin{proof}
Applying \eqref{eq35} we obtain
\begin{align*}
(\iscript ^\mscript\circ\iscript ^{\mscript '})_{(x,y)}(\rho )
   &=\trace _{K'}\brac{\nu '\sqbrac{\trace _K\paren{\sum _{x'}\alpha _{x'}(\rho )\otimes\gamma _{x'}I\otimes P_x}}I\otimes P'_y}\\
   &=\sum _{x'}\trace _{K'}\brac{\nu '\sqbrac{\trace _K\paren{\alpha _{x'}(\rho )\otimes\gamma _{x'}P_x}}I\otimes P'_y}\\
   &=\sum _{x'}\trace _{K'}\brac{\nu '\sqbrac{\trace (\gamma _{x'}P_x)\alpha _{x'}(\rho )}I\otimes P'_y}\\
   &=\sum _{x'}\trace (\gamma _{x'}P_x)\trace _{K'}\brac{\nu '\sqbrac{\alpha _{x'}(P)}I\otimes P'_y}\\
   &=\sum _{x'}\trace (\gamma _{x'}P_x)\trace _{K'}\sqbrac{\sum _{y'}\alpha '_{y'}\paren{\alpha _{x'}(\rho )}\otimes\gamma '_{y'}I\otimes P'_y}\\
   &=\sum _{x',y'}\trace (\gamma _{x'}P_x)\trace _{K'}\sqbrac{\alpha '_{y'}\paren{\alpha _{x'}(\rho )}\otimes\gamma '_{y'}P'_y}\\
   &=\sum _{x',y'}\trace (\gamma _{x'}P_x)\trace (\gamma '_{y'}P'_y)\alpha '_{y'}\paren{\alpha _{x'}(\rho )}\\
   &=\sum _{x',y'}\trace (\gamma _{x'}P_x)\trace (\gamma '_{y'}P'_y)(\alpha\circ\alpha ')_{(x',y')}(\rho )
\end{align*}
and \eqref{eq33} follows. Applying \eqref{eq31} gives \eqref{eq34}
\end{proof}

Let $\nu\colon\sscript (H)\to\sscript (H\otimes K)$ and $\nu '\colon\sscript (H)\to\sscript (H\otimes K')$ be channels. Motivated by \eqref{eq32}, we say that a channel
$\mu\colon\sscript (H)\to\sscript (H\otimes K\otimes K')$ is a \textit{product channel} for $\nu ,\nu'$ if
\begin{equation}                % equation (3.5)
\label{eq35}
\trace _{K\otimes K'}\sqbrac{\mu (\rho )I\otimes a\otimes b}=\trace _{K'}\brac{\nu '\sqbrac{\trace _K\paren{\nu (\rho )I\otimes a}}I\otimes b}
\end{equation}
for all $a\in\escript (K)$, $b\in\escript (K')$, $\rho\in\sscript (H)$.

\begin{thm}    % Theorem 3.2
\label{thm32}
Let $\nu (\rho )=\sum\limits _y\alpha _y(\rho )\otimes\gamma _y$, $\nu '(\rho )=\sum\limits _{y'}\alpha '_{y'}(\rho )\otimes\gamma '_{y'}$ be separable.
{\rm{(i)}}\enspace We have that $\mu$ is a product channel for $\nu ,\nu'$ if and only if for all $a\in\escript (K)$, $b\in\escript (K')$, $\rho\in\sscript (H)$ we have
\begin{equation}                % equation (3.6)
\label{eq36}
\trace _{K\otimes K'}\sqbrac{\mu (\rho )I\otimes a\otimes b}=\sum _{x',y'}\trace (\gamma _{x'}a)\trace (\gamma '_{y'}b)\alpha '_{y'}(\alpha _{x'}(\rho ))
\end{equation}
{\rm{(ii)}}\enspace We have that $\mu$ is a product channel for $\nu ,\nu '$ if
\begin{equation}                % equation (3.7)
\label{eq37}
\mu (\rho )=\sum _{x',y'}\alpha '_{y'}\sqbrac{\alpha _{x'}(\rho )}\otimes\gamma _{x'}\otimes\gamma '_{y'}
\end{equation}
for all $\rho\in\sscript (H)$.
\end{thm}
\begin{proof}
(i)\enspace As in the proof of Theorem~\ref{thm31} we have
\begin{align*}
\trace _{K'}\brac{\nu '\sqbrac{\trace _K\paren{\nu (\rho )I\otimes a}I\otimes b}}
   &=\trace _{K'}\brac{\nu '\sqbrac{\trace _K\paren{\sum _{x'}\alpha _{x'}(\rho )\otimes\gamma _{x'}I\otimes a}}I\otimes b}\\
   &=\sum _{x',y'}\trace (\gamma _{x'}a)\trace (\gamma '_{y'}b)\alpha '_{y'}\paren{\alpha _{x'}(\rho )}
\end{align*}
Hence, \eqref{eq35} holds if and only if \eqref{eq36} holds.
(ii)\enspace If \eqref{eq37} holds, then 
\begin{align*}
\trace _{K\otimes K'}\sqbrac{\mu (\rho )I\otimes a\otimes b}
   &=\trace _{K\otimes K'}\brac{\sum _{x',y'}\alpha '_{y'}\sqbrac{\alpha _{x'}(\rho )}\otimes\gamma _{x'}\otimes\gamma '_{y'}I\otimes a\otimes b}\\
   &=\sum _{x',y'}\trace _{K\otimes K'}\brac{\alpha '_{y'}\sqbrac{\alpha _{x'}(\rho )}\otimes\gamma _{x'}a\otimes\gamma '_{y'}b}\\
   &=\sum _{x',y'}\trace (\gamma _{x'}a)\trace (\gamma '_{y'}b)\alpha '_{y'}\paren{\alpha _{x'}(\rho )}
\end{align*}
for all $a\in\escript (K)$, $b\in\escript (K')$, $\rho\in\sscript (H)$. Hence, \eqref{eq36} holds so $\mu$ is a product channel for $\nu ,\nu '$.
\end{proof}

If $\mu\colon\sscript (H)\to\sscript (H\otimes K\otimes K')$ satisfies \eqref{eq37}, we see that $\mu$ is separable and by Theorem~\ref{thm32}(ii), $\mu$ is a product channel for $\nu ,\nu '$ and we write $\mu =\nu\times\nu '$. Let $\mscript =(K,\nu ,P)$, $\mscript '=(K',\nu ',P')$ be MMs for $H$ and let
$\mu\colon\sscript (H)\to\sscript (H\otimes K\otimes K')$ be a product channel for $\nu ,\nu '$. Then the MM given by $\mscript ''=(K\otimes K',\mu, P\otimes P')$ is called a \textit{sequential product} of $\mscript$ and $\mscript '$. In $\mscript ''$ we interpret $P\otimes P'$ as the observable $(P\otimes P')_{(x,y)}=P_x\otimes P'_y$. Applying \eqref{eq31}, the instrument measured by $\mscript ''$ becomes
\begin{align}                 % equation (3.8)
\label{eq38}
\iscript _{(x,y)}^{\mscript ''}(\rho )&=\trace _{K\otimes K'}\sqbrac{\mu (\rho )I\otimes (P\otimes P')_{(x,y)}}
   =\trace _{K\otimes K'}\sqbrac{\mu (\rho )I\otimes P_x\otimes P_{x'}}\notag\\
   &=\trace _{K'}\brac{\nu '\sqbrac{\trace _K\paren{\nu (\rho )I\otimes P_x}I\otimes P'_y}}=(\iscript ^\mscript\circ\jscript ^{\mscript '})_{(x,y)}(\rho )
\end{align}
We conclude that the instrument measured by $\mscript ''$ is the sequential product of their respective measured instruments. If $\nu$ and $\nu '$ are separable and $\mu$ satisfies \eqref{eq37}, we write $\mscript ''=\mscript\circ\mscript '$.

If $\nu\colon\sscript (H)\to\sscript (H\otimes K)$, $\nu '\colon\sscript (H)\to\sscript (H\otimes K')$ are channels, we define the channel
$(\nu '\mid\nu )\colon\sscript (H)\to\sscript (H\otimes K')$ to be
\begin{equation*}
(\nu '\mid\nu )(\rho )=\nu '\sqbrac{\trace _K(\nu (\rho ))}
\end{equation*}
and call $(\nu '\mid\nu )$ the channel $\nu '$ \textit{conditioned by} $\nu$. For MMs $\mscript =(K,\nu ,P)$, $\mscript '=(K',\nu ',P')$ we define the MM called
$\mscript '$ \textit{conditioned by} $\mscript$ to be
\begin{equation*}
\mscript '\mid\mscript =(K',(\nu '\mid\nu ),P')
\end{equation*}
Notice that in this definition, the probe observable $P$ is lost. If $\iscript ,\jscript =\iscript (H)$ we define $\jscript\mid\iscript\in\iscript (H)$ by
$\Omega _{\jscript\mid\iscript}=\Omega _\jscript$ and $(\jscript\mid\iscript )_y(\rho )=\jscript _y\paren{\iscriptbar (\rho )}$. We interpret $\jscript\mid\iscript$ as the instrument obtained from first measuring $\iscript$, disregarding the result and then measuring $\jscript$. Part (ii) of the next theorem unites the concepts of conditioned and sequential products of MMs.

\begin{thm}    % Theorem 3.3
\label{thm33}
{\rm{(i)}}\enspace If $\nu (\rho )=\sum\alpha _x(\rho )\otimes\gamma _x$, $\nu '(\rho )=\sum\alpha '_y(\rho )\otimes\gamma '_y$ are separable then $(\nu '\mid\nu )$ is separable and
\begin{equation*}
(\nu '\mid\nu )(\rho )=\sum _y\alpha '_y\paren{\alphabar (\rho )}\otimes\gamma '_y               
\end{equation*}
{\rm{(ii)}}\enspace If $\mscript =(K,\nu ,P)$, $\mscript '(K',\nu ',P')$ are MMs, then $\iscript ^{\mscript '\mid\mscript}=(\iscript ^{\mscript '}\mid\iscript ^{\mscript})$.
Moreover, if $\mscript ''$ is a sequential product of $\mscript$ and $\mscript '$, then 
\begin{equation*} 
(\iscript ^{\mscript '}\mid\iscript ^{\mscript})_y(\rho )=\sum _x(\iscript ^{\mscript}\circ\iscript ^{\mscript '})_{(x,y)}(\rho )=\sum _x\iscript _{(x,y)}^{\mscript ''}(\rho )     
\end{equation*}
for every $\rho\in\sscript (H)$.
\end{thm}
\begin{proof}
(i)\enspace For all $\rho\in\sscript (H)$ we have
\begin{align*}
(\nu '\mid\nu )(\rho )=\nu '\sqbrac{\trace _K\paren{\nu (\rho )}}&=\nu '\sqbrac{\trace _K\paren{\sum _x\alpha _x(\rho )\otimes\gamma _x}}
   =\sum _x\nu '\sqbrac{\trace _K\paren{\alpha _x(\rho )\otimes\gamma _x}}\\
   &=\sum _x\nu '\paren{\alpha _x(\rho )}=\nu '\paren{\alphabar (\rho )}=\sum _y\alpha '_y\paren{\alphabar (\rho )}\otimes\gamma '_y
\end{align*}
(ii)\enspace For all $\rho\in\sscript (H)$, $y\in\Omega _{\iscript ^{\mscript '}}$ we have
\begin{align}                 % equation (3.9)
\label{eq39}
(\iscript ^{\mscript '\mid\mscript})_y(\rho )&=\trace _{K'}\brac{(\nu '\mid\nu )(\rho )I\otimes P'_y}
   =\trace _{K'}\brac{\nu '\sqbrac{\trace _K\paren{\nu (\rho )}}I\otimes P'_y}\notag\\
   &=\iscript _y^{\mscript '}\sqbrac{\trace _K\paren{\nu (\rho )}}=\iscript _y^{\mscript '}\paren{\iscriptbar ^{\mscript}(\rho )}
   =(\iscript ^{\mscript '}\mid\iscript ^\mscript )_y(\rho )
\end{align}
and the result follows. Applying \eqref{eq38} and \eqref{eq39} we have
\begin{align*}
\sum _x\iscript _{(x,y)}^{\mscript ''}(\rho )&=\sum _x(\iscript ^{\mscript}\circ\iscript ^{\mscript '})_{(x,y)}(\rho )
   =\trace _{K'}\brac{\nu '\sqbrac{\trace _K\paren{\nu (\rho )I\otimes\sum _xP_x}}I\otimes P'_y}\\
   &=\trace _{K'}\brac{\nu '\sqbrac{\trace _K\paren{\nu (\rho )I\otimes I}}I\otimes P'_y}\\
   &=\trace _{K'}\brac{\nu '\sqbrac{\trace _K\paren{\nu (\rho )}}I\otimes P'_y}=(\iscript ^{\mscript '}\mid\iscript ^\mscript )_y(\rho )\qedhere
\end{align*}
\end{proof}

\begin{exam}{3}  % Example 3
Let $\nu (\rho )=\sum\limits _y\alpha _y(\rho )\otimes\gamma _y$, $\nu '(\rho )=\sum _{y'}\alpha '_{y'}(\rho )\gamma '_{y'}$ be separable channels with
$\alpha _y(\rho )=C_y\rho C_y^*$, $\alpha '_{y'}(\rho )=D_{y'}\rho D_{y'}^*$ being Kraus instruments. Define the product channel $\nu\times\nu '$ for $\nu ,\nu '$ according to Theorem~\ref{thm32}(ii) by
\begin{equation*}
\nu\times\nu '(\rho )=\sum _{x',y'}\alpha '_{y'}\sqbrac{\alpha _{x'}(\rho )}\otimes\gamma _{x'}\otimes\gamma '_{y'}
   =\sum _{x',y'}D_{y'}C_{x'}\rho C_{x'}^*D_{y'}^*\otimes\gamma _{x'}\otimes\gamma '_{y'}
\end{equation*}
If $\mscript =(K,\nu ,P)$, $\mscript '=(K',\nu ',P')$, we have the sequential product MM
\begin{equation*}
\mscript\circ\mscript '=(K\otimes K',\nu\times\nu ',P\otimes P')
\end{equation*}
The instrument measured by $\mscript\circ\mscript '$ according to \eqref{eq38} and Theorem~\ref{thm31} becomes
\begin{align*}
\iscript _({x,y)}^{\mscript\circ\mscript '}(\rho )&=(\iscript\circ\iscript ^{\mscript '})_{(x,y)}(\rho )
   =\sum _{x',y'}\trace (\gamma _{x'}P_x)\trace (\gamma '_{y'}P'_y)\alpha '_{y'}\sqbrac{\alpha _{x'}(\rho )}\\
   &=\sum _{x',y'}\trace (\gamma _{x'}P_x)\trace (\gamma '_{y'}P'_y)D_{y'}C_{x'}\rho C_{x'}^*D_{y'}^*
\end{align*}
By Theorem~\ref{thm31}, the observable measured by $\mscript\circ\mscript '$ is
\begin{equation*}
(\iscript ^{\mscript\circ\mscript '}_{(x,y)})^\wedge=\sum _{x',y'}\trace (\gamma _{x'}P_x)\trace (\gamma '_{y'}P'_y)C_{x'}^*D_{y'}^*C_{y'}C_{x'}
\end{equation*}
From Theorem~\ref{thm33}(i) the channel $(\nu '\mid\nu )$ becomes
\begin{equation*}
(\nu '\mid\nu )(\rho )=\sum _{y,y'}\alpha '_{y'}\paren{\alpha _y(\rho )}\otimes\gamma '_{y'}=\sum _{y,y'}D_{y'}C_y\rho C_y^*D_{y'}^*\otimes\gamma '_{y'}
\end{equation*}
We also have
\begin{align*}
\iscript _y^{\mscript '\mid\mscript}(\rho )&=(\iscript ^{\mscript '}\mid\iscript ^\mscript )_y(\rho )=\iscript _y^{\mscript '}\paren{\iscriptbar ^\mscript (\rho )}
   =\iscript _y^{\mscript '}\paren{\sum _xC_x\rho C_x^*}\\
   &=\sum _x\iscript _y^{\mscript '}(C_x\rho C_x^*)=\sum _{x,y'}\trace (\gamma '_{y'}P'_y)D_{y'}C_x\rho C_x ^*D_{y'}^*
\end{align*}
Moreover,
\begin{equation*}
(\iscript ^{\mscript '\mid\mscript})_y^\wedge =\sum _{x,y'}\trace (\gamma '_{y'}P'_y)C_x^*D_{y'}^*D_{y'}C_x
\end{equation*}
In particular, if $\alpha _y(\rho )=A_y^{1/2}\rho A_y^{1/2}$, $\alpha '_{y'}(\rho )=B_{y'}^{1/2}\rho B_{y'}^{1/2}$ are L\"uders instruments, the formulas are similar. For example, 
\begin{align*}
(\iscript ^{\mscript\circ\mscript '})_{(x,y)}^\wedge&=\sum _{x',y'}\trace (\gamma _{x'}P_x)\trace (\gamma '_{y'}P'_y)A_x^{1/2}B_{y'}A_{x'}^{1/2}\\
   (\iscript ^{\mscript '\mid\mscript})_y^\wedge&=\sum _{x',y'}\trace (\gamma '_{y'}P'_y)A_x^{1/2}B_{y'}A_x^{1/2}\hskip 8pc\qedsymbol
\end{align*}
\end{exam}

\begin{exam}{4}  % Example 4
Let $\nu (\rho )=\sum\limits _y\alpha _y(\rho )\otimes\gamma _y$, $\nu '(\rho )=\sum\limits _{y'}\alpha '_{y'}(\rho )\otimes\gamma '_{y'}$ be separable channels with
$\alpha _y(\rho )=\trace (\rho A_y)\beta _y$, $\alpha '_{y'}(\rho )=\trace (\rho B_{y'})\beta '_{y'}$ being Holevo instruments. Define the separable product channel
$\nu\times\nu '$ for $\nu ,\nu '$ according to Theorem~\ref{thm32}(ii) by
\begin{align*}
\nu\times\nu '&=\sum _{x',y'}\alpha '_{y'}\sqbrac{\alpha _{x'}(\rho )}\otimes\gamma _{x'}\otimes\gamma '_{y'}
   =\sum _{x',y'}\alpha '_{y'}\sqbrac{\trace (\rho A_{x'})\beta _{x'}}\otimes\gamma _{x'}\otimes\gamma '_{y'}\\
   &=\sum _{x',y'}\trace (\rho A_{x'})\alpha '_{y'}(\beta _{x'})\otimes\gamma _{x'}\otimes\gamma '_{y'}\\
   &=\sum _{x',y'}\trace (\rho A_{x'})\trace (\beta _{x'}B_{y'})\beta '_{y'}\otimes\gamma _{x'}\otimes\gamma '_{y'}
\end{align*}
Letting $\mscript =(K,\nu ,P)$, $\mscript '=(K',\nu ',P')$ we have the sequential product MM given by
\begin{equation*}
\mscript\circ\mscript '=(K\otimes K',\nu\times\nu ',P\otimes P')
\end{equation*}
The instrument measured by $\mscript\circ\mscript '$ becomes by \eqref{eq38} and Theorem~\ref{thm31}
\begin{align*}
\iscript _{(x,y)}^{\mscript\circ\mscript '}(\rho )&=(\iscript ^\mscript\circ\iscript ^{\mscript '})_{(x,y)}(\rho )
   =\sum _{x',y'}\trace (\gamma _{x'}P_x)\trace (\gamma '_{y'}P'_y)\alpha '_{y'}\sqbrac{\alpha _{x'}(\rho )}\\
   &=\sum _{x',y'}\trace (\gamma _{x'}P_x)\trace (\gamma '_{y'}P'_y)\alpha '_{y'}\sqbrac{\trace (\rho A_{x'})\beta _{x'}}\\
   &=\sum _{x',y'}\trace (\gamma _{x'}P_x)\trace (\gamma '_{y'}P'_y)\trace (\rho A_{x'})\alpha '_{y'}(\beta _{x'})\\
   &=\sum _{x',y'}\trace (\gamma _{x'}P_x)\trace (\gamma '_{y'}P'_y)\trace (\rho A_{x'})\trace (\beta _{x'}B_{y'})\beta '_{y'}
\end{align*}
By Theorem~\ref{thm31}, the observable measured by $\mscript\circ\mscript '$ is
\begin{equation*}
(\iscript ^{\mscript\circ\mscript '})_{(x,y)}^\wedge =\sum _{x',y'}\trace (\gamma _{x'}P_x)\trace (\gamma '_{y'}P'_y)\trace (\beta _{x'}B_{y'})A_{x'}
\end{equation*}
From Theorem~\ref{thm33}(i) the channel $(\nu '\mid\nu )$ becomes
\begin{align*}
(\nu '\mid\nu )(\rho )&=\sum _{y,y'}\alpha '_{y'}\paren{\alpha _y(\rho )}\otimes\gamma '_{y'}
   =\sum _{y,y'}\alpha '_{y'}\sqbrac{\trace (\rho A_y)\beta _y}\otimes\gamma '_{y'}\\
   &=\sum _{y,y'}\trace (\rho A_y)\alpha '_{y'}(\beta _y)\otimes\gamma '_{y'}=\sum _{y,y'}\trace (\rho A_y)\trace (\beta _yB_{y'})\beta '_{y'}\otimes\gamma '_{y'}
\end{align*}
We also have
\begin{align*}
\iscript _y^{\mscript '\mid\mscript}(\rho )&=\iscript _y^{\mscript '}\sqbrac{\iscriptbar ^\mscript (\rho )}
   =\iscript _y^{\mscript '}\sqbrac{\sum _{y'}\trace (\rho A_{y'})\beta _{y'}}\\
   &=\sum _{y'}\trace (\rho A_{y'})\iscript _y^{\mscript '}(\beta _{y'})=\sum _{y'}\trace (\rho A_{y'})\sum _x\trace (\gamma '_xP'_y)\trace (\beta _{y'}B_x)\beta '_x\\
   &=\sum _{x,y'}\trace (\rho A_{y'})\trace (\gamma '_xP'_y)\trace (\beta _{y'}B_x)\beta '_x
\end{align*}
Moreover,
\begin{equation*}
(\iscript ^{\mscript '\mid\mscript})_y^\wedge =\sum _{x,y'}\trace (\gamma '_xP'_y)\trace (\beta _{y'}B_x)A_{y'}
\end{equation*}

As a special case, let $A_y=\lambda _yI$, $B_{y'}=\mu _{y'}I$ be identity observables and $\alpha _y(\rho )=\trace (\rho A_y)\beta =\lambda _y\beta$,
$\alpha '_{y'}(\rho )=\trace (\rho B_{y'})\beta '=\mu _{y'}\beta '$ be corresponding trivial instruments. Then
\begin{equation*}
\nu\times\nu '(\rho )=\sum _{x',y'}\lambda _{x'}\mu _{y'}\beta '\otimes\gamma _{x'}\otimes\gamma '_{y'}=\beta '\otimes\gamma\otimes\gamma '
\end{equation*}
where $\gamma =\sum\limits _{x'}\lambda _{x'}\gamma _{x'}\in\sscript (H)$, $\gamma '=\sum\limits _{y'}\mu _{y'}\gamma '_{y'}\in\sscript (K')$. We then obtain
\begin{equation*}
\iscript _{(x,y)}^{\mscript\circ\mscript '}(\rho )=\sum _{x',y'}\trace (\gamma _{x'}P_x)\trace (\gamma '_{y'}P'_y)\lambda _{x'}\mu _{y'}\beta '
   =\trace (\gamma P_x)\trace (\gamma 'P'_y)\beta '
\end{equation*}
which is a trivial instrument. We also have
\begin{equation*}
(\iscript ^{\mscript\circ\mscript '})_{(x,y)}^\wedge =\trace (\gamma P_x)\trace (\gamma 'P'_y)I
\end{equation*}
which is an identity observable. From Theorem~\ref{thm33}(i), $(\nu '\mid\nu )$ becomes the constant channel
\begin{equation*}
(\nu '\mid\nu )(\rho )=\sum _{y,y'}\lambda _y\mu _{y'}\beta '\otimes\gamma '_{y'}=\beta '\otimes\gamma '
\end{equation*}
We also have
\begin{equation*}
\iscript _y^{\mscript '\mid\mscript}(\rho )=\sum _{x,y'}\lambda _{y'}\trace (\gamma '_xP'_y)\mu _x\beta '=\trace (\gamma 'P'_y)\beta '
\end{equation*}
which is a trivial instrument. Finally,
\begin{equation*}
(\iscript ^{\mscript '\mid\mscript})_y^\wedge =\trace (\gamma 'P'_y)I
\end{equation*}
which is an identity observable.\hskip 16pc\qedsymbol
\end{exam}

\section{Measurement Model Statistics}  % Section 4
An observable $A=\brac{A_x\colon x\in\Omega _A}$ is \textit{real-valued} if $\Omega _A\subseteq\real$ \cite{gud23,hz12}. When $A$ is real-valued, the
\textit{stochastic operator for} $A$ is the self-adjoint operator $\atilde =\sum\limits _xxA_x$ \cite{gud23}. In a similar way, an instrument $\iscript$ is
\textit{real-valued} if $\Omega _\iscript\subseteq\real$ and the \textit{stochastic map for} $\iscript$ is $\iscripttilde =\sum\limits _xx\iscript _x$. A MM given by
$\mscript =\brac{K,\nu ,P}$ is \textit{real-valued} if $P$ is real-valued and we define the \textit{stochastic operator for} $\mscript$ to be
$\mscripttilde =\sqbrac{(\iscript ^\mscript )^\wedge}^\sim$. If $B\in\lscript _S(H)$, $\rho\in\sscript (H)$, the $\rho$-\textit{expectation} (or $\rho$-\textit{average value}) of $B$ is $E_\rho (B)=\trace (\rho B)$. For a real-valued observable $A$ or a real-valued instrument $\iscript$ and $\rho\in\sscript (H)$ we define
\begin{align*}
E_\rho (A)&=E_\rho (\atilde\,)=\trace (\rho\atilde\,)=\sum _xx\trace (\rho A_x)\\
\intertext{and}
E_\rho (\iscript )&=\trace\sqbrac{\iscripttilde (\rho )}=\sum _xx\trace\sqbrac{\iscript _x(\rho )}=\sum _xx\trace (\rho\iscripthat _x)=E_\rho (\iscripthat\,)
\end{align*}
If $\mscript$ is a real-valued MM, we define
\begin{equation*}
E_\rho (\mscript )=E_\rho (\mscripttilde )=\trace (\rho\mscripttilde )=\sum _xx\trace\sqbrac{\rho (\iscript ^\mscript )_x^\wedge}
   =\sum _xx\trace\sqbrac{\iscript _x^\mscript (\rho )}
\end{equation*}
For $B\in\lscript _S(H)$, $\rho\in\sscript (H)$, the $\rho$-\textit{variance} of $B$ is
\begin{equation*}
\Delta _\rho (B)=E_\rho\sqbrac{\paren{B-E_\rho (B)I}^2}=E_\rho (B^2)-E_\rho (B)^2=\trace (\rho B^2)-\sqbrac{\trace (\rho B)}^2
\end{equation*}
We define $\Delta _\rho (\mscript )=\Delta _\rho (\,\mscripttilde )$.

\begin{thm}    % Theorem 41
\label{thm41}
{\rm{(i)}}\enspace If $\mscript =(K,\nu ,P)$ is a real-valued MM, then
\begin{align*}
E_\rho (\mscript )&=\trace\sqbrac{\nu (\rho )I\otimes\ptilde}\\
\Delta _\rho (\mscript )&=\trace\brac{\nu\sqbrac{\rho (\iscript ^\mscript )^\sim}I\otimes\ptilde}-\brac{\trace\sqbrac{\nu (\rho )I\otimes\ptilde}}^2
\end{align*}
{\rm{(ii)}}\enspace If $\nu (\rho )=\sum _y\alpha _y(\rho )\otimes\gamma _y$ is separable, then
\begin{align*}
E_\rho (\mscript )&=\sum _y\trace\sqbrac{\alpha _y(\rho )}\trace (\gamma _y\ptilde ),\quad\mscripttilde =\sum _y\trace (\gamma _y\ptilde )\alphahat _y\\
\Delta _\rho (\mscript )&=\sum _{y,y'}\trace (\gamma _y\ptilde )\trace (\gamma _{y'}\ptilde )
   \sqbrac{\trace (\rho\alphahat _y\alphahat _{y'})-\trace (\rho\alphahat _y)\trace (\rho\alphahat _{y'})}
\end{align*}
\end{thm}
\begin{proof}
(i)\enspace For every $\rho\in\sscript (H)$ we have
\begin{align*}
E_\rho (\mscript )&=\sum _xx\trace\sqbrac{\iscript _x^\mscript (\rho )}=\sum _xx\trace\sqbrac{\nu (\rho )I\otimes P_x}\\
    &=\trace\sqbrac{\nu (\rho )I\otimes\sum _xP_x}=\trace\sqbrac{\nu (\rho )\otimes\ptilde}
\end{align*}
Moreover, we have
\begin{align*}
\trace (\rho\mscripttilde\,^2)&=\trace\sqbrac{\rho\sum _x(\iscript ^\mscript )_x^\wedge\sum _yy(\iscript ^\mscript )_y^\wedge}
   =\sum _{x,y}xy\sqbrac{\rho (\iscript ^\mscript )_x^\wedge (\iscript ^\mscript )_y^\wedge}\\
   &=\sum _{x,y}xy\trace\sqbrac{\iscript _y^\mscript\paren{\rho (\iscript ^\mscript )_x^\wedge}}
   =\sum _yy\sqbrac{\iscript _y^\mscript\paren{\rho (\iscript ^\mscript )^\sim}}\\
   &=\sum _yy\trace\sqbrac{\nu\paren{\rho (\iscript ^\mscript )^\sim}I\otimes P_y}=\trace\sqbrac{\nu\paren{\rho (\iscript ^\mscript )^\sim}I\otimes\ptilde}
\end{align*}
and the result follows.\newline
(ii)\enspace Applying (i) gives
\begin{align*}
E_\rho (\mscript )&=\trace\sqbrac{\sum _y\alpha _y(\rho )\otimes\gamma _yI\otimes\ptilde}=\sum _y\trace\sqbrac{\alpha _y(\rho )\otimes\gamma _y\ptilde}\\
   &=\sum _y\trace\sqbrac{\alpha _y(\rho )}\trace (\gamma _y\ptilde )
\end{align*}
Since $(\iscript ^\mscript )_x^\wedge =\sum _y\trace (\gamma _yP_x)\alphahat _y$ we obtain
\begin{equation*}
\mscripttilde =\sum _xx(\iscript ^\mscript )_x^\wedge =\sum _y\trace (\gamma _y\ptilde )\alphahat _y
\end{equation*}
Applying \eqref{eq23} we have $\sqbrac{(\iscript ^\mscript )^\wedge}^\sim=\sum\trace (\gamma _y\ptilde )\alphahat _y$. Hence,
\begin{equation*}
\brac{\sqbrac{(\iscript ^\mscript )^\wedge}^\sim}^2=\sum _{y,y'}\trace (\gamma _y\ptilde )\trace (\gamma _{y'}\ptilde )\alphahat _y\alphahat _{y'}
\end{equation*}
It follows that
\begin{align*}
\Delta _\rho (\mscript )&=\sum _{y,y'}\trace (\gamma _y\ptilde )\trace (\gamma _{y'}\ptilde )\trace (\rho\alphahat _y\alphahat _{y'})
   -\sum _{y,y'}\trace (\gamma _y\ptilde )\trace (\gamma _{y'}\ptilde )\trace (\rho\alphahat _y)\trace (\rho\alphahat _{y'})\\
   &=\sum _{y,y'}\trace (\gamma _y\ptilde )\trace (\gamma _{y'}\ptilde )
   \sqbrac{\trace (\rho\alphahat _y\alphahat _{y'})-\trace (\rho\alphahat _y)\trace (\rho\alphahat _{y'})}\qedhere
\end{align*}
\end{proof}

For MMs given by $\mscript ,\mscript '$ and $\rho\in\sscript (H)$ we define their $\rho$-\textit{correlation} as
\begin{equation*}
\rmcor _\rho (\mscript ,\mscript ')=\trace (\rho\mscripttilde{\mscript '}^\sim)-E_\rho (\mscripttilde )E_\rho ({\mscript '}^\sim )
\end{equation*}
their $\rho$-\textit{covariance} as $\Delta _\rho (\mscript ,\mscript  ')=\rmre\rmcor _\rho (\mscript ,\mscript ')$ and their $\rho$-\textit{commutator term} as
\begin{equation*}
\rmcomm _\rho (\mscript ,\mscript ')=\tfrac{1}{4}\ab{\trace\paren{\rho\sqbrac{\mscripttilde ,{\mscript '}^\sim}}}^2
\end{equation*}
Note that $\Delta _\rho (\mscript )=\Delta _\rho (\mscript ,\mscript )$. The following \textit{uncertainty principle} relates these various concepts.

\begin{thm}    % Theorem 4.2
\label{thm42}
\rm{\cite{gud23}}\enspace
$\rmcomm _\rho (\mscript ,\mscript ')+\sqbrac{\Delta _\rho (\mscript ,\mscript ')}^2
   =\ab{\rmcor _\rho (\mscript ,\mscript ')}^2\le\Delta _\rho (\mscript )\Delta _\rho (\mscript ')$.
\end{thm}

We now compute the terms in Theorem~\ref{thm42} when $\nu (\rho )=\sum\limits _y\alpha _y(\rho )\otimes\gamma _y$ and
$\nu '(\rho )=\sum\limits _{y'}\alpha '_{y'}(\rho )\otimes\gamma '_{y'}$ are separable. We have that $\Delta _\rho (\mscript )$, $\Delta _\rho (\mscript ')$ are given in
Theorem~\ref{thm41}(ii). Also, applying Theorem~\ref{thm41}(ii) we obtain 
\begin{align*}
\sqbrac{\mscripttilde ,\mscripttilde '}&=\sqbrac{\sum _y\trace (\gamma _y\ptilde )\alphahat _y,\sum _{y'}\trace (\gamma '_{y'}\ptilde ')\alphahat\,'_{y'}}\\
   &=\sum _{y,y'}\trace (\gamma _y\ptilde )\trace (\gamma '_{y'}\ptilde ')\sqbrac{\alphahat _y,\alphahat\,'_{y'}}
\end{align*}
Hence,
\begin{equation*}
\trace\paren{\rho\sqbrac{\mscripttilde ,\mscripttilde '}}=\sum _{y,y'}\trace (\gamma _y\ptilde )\trace (\gamma '_{y'}\ptilde\,')
   \sqbrac{\trace (\rho\alphahat _y\alphahat\,'_{y'})-\trace (\rho\alphahat '_{y'}\alphahat _y)}
\end{equation*}
and $\rmcomm _\rho (\mscript ,\mscript ')$ follows directly. The $\rho$-correlation becomes
\begin{align*}
\rmcor _\rho (\mscript ,\mscript ')&=\trace\paren{\rho\sum _y\trace (\gamma _y\ptilde )\alphahat _y\sum _{y'}\trace (\gamma '_{y'}\ptilde\,')\alphahat\,'_{y'}}\\
   &\quad -\sum _{y,y'}\trace (\gamma _y\ptilde )\trace (\rho\alphahat _y)\trace (\gamma '_{y'}\ptilde\,')\trace (\rho\alphahat\,'_{y'})\\
   &=\sum _{y,y'}\trace (\gamma _y\ptilde )\trace (\gamma '_{y'}\ptilde\,')
   \sqbrac{\trace (\rho\alphahat _y\alphahat '_{y'})-\trace (\rho\alphahat _y)\trace (\rho\alphahat\,'_{y'})}
\end{align*}
The $\rho$-covariance follows.

\begin{exam}{5}  % Example 5
As in Example~1, let $\nu (\rho )=\sum\limits _{y}\alpha _y(\rho )\otimes\gamma _y$, $\nu '(\rho )=\sum\limits _{y'}\alpha '_{y'}\otimes\gamma '_{y'}$
be separable channels with $\alpha _y(\rho )=A_y^{1/2}\rho A_y^{1/2}$, $\alpha '_{y'}(\rho )=B_{y'}^{1/2}\rho B_{y'}^{1/2}$ being L\"uders instruments and let
$\mscript =(K,\nu ,P)$, $\mscript '=(K',\nu ',P')$. Applying Theorem~\ref{thm41}(ii) we obtain
\begin{align*}
E_\rho (\mscript )&=\sum _y\trace (\rho A_y)\trace (\gamma _y\ptilde ),\quad \mscripttilde =\sum _y\trace (\gamma _y\ptilde )A_y\\
   \Delta _\rho (\mscript )&=\sum _{y,y'}\trace (\gamma _y\ptilde )\trace (\gamma _{y'}\ptilde )
   \sqbrac{\trace (\rho A_yA_{y'})-\trace (\rho A_y)\trace (\rho A_{y'})}
\end{align*}
with similar expressions for $E_\rho (\mscript ')$, $(\mscript ')^\sim$ and $\Delta _\rho (\mscript ')$. Moreover,
\begin{align*}
\trace\paren{\rho\sqbrac{\mscripttilde ,\mscripttilde\,'}}
   &=\sum  _{y,y'}\trace (\gamma _y\ptilde )\trace (\gamma '_{y'}\ptilde ')\brac{\trace\paren{\rho\sqbrac{A_y,B_{y'}}}}\\
   \rmcor _\rho (\mscript ,\mscript ')&=\sum _{y,y'}\trace (\gamma _y\ptilde )\trace (\gamma '_{y'}\ptilde ')
   \sqbrac{\trace (\rho A_yB_{y'})-\trace (\rho A_y)\trace (\rho B_{y'})}\hfill\square
\end{align*}
\end{exam}

\begin{exam}{6}  % Example 6
As in Example~2 and 4, let $\nu (\rho )=\sum\limits _y\alpha _y(\rho )\otimes\gamma _y$, $\nu '(\rho )=\sum\limits _{y'}\alpha '_{y'}(\rho )\otimes\gamma '_{y'}$ be separable channels with $\alpha _y(\rho )=\trace (\rho A_y)\beta _y$, $\alpha '_{y'}(\rho )=\trace (\rho B_{y'})\beta '_y$ being Holevo instruments and let
$\mscript =(K,\nu ,P)$, $\mscript '=(K',\nu ',P')$. Even though the instruments are different than those in Example~5, all the statistical expressions are the same. As a special case, let $\alpha _y=\lambda _y\beta$ be a trivial instrument. Then $\alphahat _y=\lambda _yI$ is an identity observable. We then have
\begin{align*}
E_\rho (\mscript )&=\sum _y\lambda _y\trace (\gamma _y\ptilde )=\trace (\gamma\ptilde )\\
\mscripttilde&=\sum _y\trace (\gamma _y\ptilde )\lambda _yI=\trace (\gamma\ptilde )I
\end{align*}
Since
\begin{equation*}
\Delta _\rho (\mscript )=\trace\paren{\rho\sqbrac{\mscripttilde ,\mscripttilde '}}=\rmcor _\rho '(\mscript ,\mscript ')=\Delta _\rho (\mscript ,\mscript ')=0
\end{equation*}
the uncertainty principle vanishes trivially\hskip 10pc\qedsymbol
\end{exam}


\begin{thebibliography}{99}
% ref 1
\bibitem{bgl95}P.~Busch, M.~Grabowski and P.~Lahti, \textit{Operational Quantum Physics}, Springer-Verlag, Berlin,
 1995.
% ref 2
\bibitem{blm96}P.~Busch, P.~Lahti and P.~Mittlestaedt, \textit{The Quantum Theory of Measurement}, Springer-Verlag, Berlin,
 1996.
% ref 3
\bibitem{blpy16}P.~Busch, P.~Lahti J.-P.~Pellonp\"a\"a and K.~Ylinen \textit{Quantum Measurement}, Springer, 2016.
% ref 4
\bibitem{gud120}S.\,Gudder, Parts and composites of quantum systems, arXiv:quant-ph 2009.07371 (2020).
% ref 5
\bibitem{gud220}----, Combinations of quantum observables and instruments, arXiv:quant-ph 2010.08025 (2020).
% ref 6
\bibitem{gud23}----, Real-valued observables and quantum uncertainty, arXiv:quant-ph 2301.07185 (2023).
% ref 7
\bibitem{hz12}T.~Heinosaari and M.~Ziman, \textit{The Mathematical Language of Quantum Theory}, Cambridge University Press, Cambridge, 2012.
% ref 8
\bibitem{lp22}P.\,Lahti and J.-P.\,Pellonp\"a\"a, An attempt to understand relational quantum mechanics, arXiv:quant-ph 2207.0138 v.2 (2022).
% ref 9
\bibitem{nc00}M.~Nielson and I.~Chuang, Quantum Computation and Quantum Information, Cambridge University Press, Cambridge, 2000.
% ref 10
\bibitem{oza84}M.\,Ozawa, Quantum measuring processes of continuous observables, \textit{J.\,Math.\,Phys.}\,\textbf{25},\,79--87\,(1984).

\end{thebibliography}
\end{document}